\documentclass[letterpaper, 10 pt, conference]{ieeeconf}
\IEEEoverridecommandlockouts                              

\usepackage{amsfonts}       

\usepackage{amsthm}
\usepackage{mathtools}      
\usepackage{amssymb}        
\usepackage{filecontents}
\usepackage{graphicx}       
\usepackage{marginnote}     
\usepackage{marvosym}       
\usepackage{overpic}        
\usepackage{tabularx}
\usepackage{pgfplots}       
\usepackage{tikz}           
\usepackage{cite}
\usepackage[linesnumbered,algoruled,boxed,lined]{algorithm2e}

\usepackage{soul}
\setstcolor{blue}

\colorlet{robot_blue}{blue!15}
\colorlet{robot_red}{red!30}
\colorlet{robot_green}{green!20}
\colorlet{robot_yellow}{yellow!20}
\def\hexagonsize{0.578cm}
\usetikzlibrary{
    patterns,
    arrows,
    arrows.meta,
    shapes,
    shapes.geometric,
    decorations.markings}
\tikzset{robot_node/.style={thick,shape=circle,draw=black,
fill=robot_blue,text=black,minimum size=0.5cm,inner sep=0.1}}
\tikzset{point_node/.style={thick,shape=circle,draw=black,
fill=black,minimum size=2,inner sep=0}}
\tikzset{vertex_node/.style={thick,shape=circle,draw=black,
fill=black,minimum size=5,inner sep=0}}
\tikzset{every picture/.style={thick, font=\small}}
\tikzset{>={Latex[length=2mm]}}
\tikzset{hexagon/.style={semithick, shape=regular polygon, regular polygon sides=6,
minimum size=2*\hexagonsize, draw=gray, inner sep=0,anchor=south}}

\pgfdeclarepatternformonly
  {hexagons}
  {\pgfpointorigin}
  {\pgfpoint{3*\hexagonsize}{0.866025*2*\hexagonsize}}
  {\pgfpoint{3*\hexagonsize}{0.866025*2*\hexagonsize}}
  {
   \pgfsetlinewidth{0.3pt}
   \pgftransformshift{\pgfpoint{0mm}{0.866025*\hexagonsize}}
   \pgfpathmoveto{\pgfpoint{0mm}{0mm}}
   \pgfpathlineto{\pgfpoint{0.5*\hexagonsize}{0mm}}
   \pgfpathlineto{\pgfpoint{\hexagonsize}{-0.866025*\hexagonsize}}
   \pgfpathlineto{\pgfpoint{2*\hexagonsize}{-0.866025*\hexagonsize}}
   \pgfpathlineto{\pgfpoint{2.5*\hexagonsize}{0mm}}
   \pgfpathlineto{\pgfpoint{3*\hexagonsize+0.2mm}{0mm}}
   \pgfpathmoveto{\pgfpoint{0.5*\hexagonsize}{0mm}}
   \pgfpathlineto{\pgfpoint{\hexagonsize}{0.866025*\hexagonsize}}
   \pgfpathlineto{\pgfpoint{2*\hexagonsize}{0.866025*\hexagonsize}}
   \pgfpathlineto{\pgfpoint{2.5*\hexagonsize}{0mm}}
   \pgfusepath{stroke}
  }

\pgfplotsset{compat=1.13}
\newtheorem{problem}{Problem}

\newtheorem{lemma}{Lemma}[section]

\newtheorem{corollary}{Corollary}[section]
\newtheorem{theorem}{Theorem}[section]
\theoremstyle{definition}
\newtheorem{definition}{Definition}[section]
\theoremstyle{remark}
\newtheorem*{remark}{Remark}

\SetKwProg{Fn}{Function}{}{}
\SetKwComment{Comment}{$\triangleright$\ }{}

\newcommand{\sag}{\textsc{SAG}\xspace}
\newcommand{\tsm}{\esrc}
\newcommand{\tsmilp}{\textsc{\esrc-ILP}\xspace}
\newcommand{\tsmsag}{\textsc{\esrc-SAG}\xspace}
\newcommand{\cmpp}{\textsc{CMPP}\xspace}
\newcommand{\mrpp}{\textsc{MPP}\xspace}

\newcommand{\esrc}{\textsc{SEAR}\xspace}

\newcommand{\rspace}{\mathbb{R}}

\newcommand{\starts}{X^I}
\newcommand{\goals}{X^G}
\newcommand{\nonneg}{\mathbb{R}_{\geq 0}}

\title{\LARGE \bf
\esrc: A Polynomial-Time Multi-Robot Path Planning Algorithm with Expected 
Constant-Factor Optimality Guarantee 
}

\author{
Shuai D. Han \quad Edgar J. Rodriguez \quad Jingjin Yu
\thanks{
S. D. Han and J. Yu are with the Department of Computer Science, 
Rutgers, the State University of New Jersey, Piscataway, NJ, USA. 
\{{\tt shuai.han, jingjin.yu}\}\hspace*{.25em}\MVAt \hspace*{.25em}rutgers.edu. 
E. J. Rodriguez is with the Department of Mathematics, 
University of California Berkeley, Berkeley, CA, USA. 
{\tt ejaramillo}\hspace*{.25em}\MVAt \hspace*{.25em}berkeley.edu. 
}%
}

\begin{document}
\maketitle

\thispagestyle{empty}
\pagestyle{empty}

\begin{abstract}
We study the labeled multi-robot path planning problem in 
continuous 2D and 3D domains in the absence of obstacles where 
robots must not collide with each other.

For an arbitrary number of robots in arbitrary initial and goal 
arrangements, we derive a polynomial time, complete 
algorithm that produces solutions with constant-factor optimality 
guarantees on both makespan and distance optimality, in expectation, 
under the assumption that the robot labels are uniformly randomly 
distributed.

Our algorithm only requires a small constant factor expansion of the 
initial and goal configuration footprints for solving the problem, 
i.e., the problem can be solved in a fairly small bounded region. 

Beside theoretical guarantees, we present a thorough computational 
evaluation of the proposed solution. In addition to the baseline 
implementation, adapting an effective (but non-polynomial time) 
routing subroutine, we also provide a highly efficient implementation 
that quickly computes near-optimal solutions. Hardware experiments 
on the microMVP platform composed of non-holonomic robots confirms the 
practical applicability of our algorithmic pipeline. 


\end{abstract}
\section{Introduction}\label{sec:Introduction}

In a labeled multi-robot path planning (\mrpp) 
problem\footnote{In this paper, path planning and motion planning 
carry identical meanings, i.e., the computation of time-parametrized,
collision-free, dynamically feasible solution trajectories. We often 
use \mrpp\ as an umbrella term to refer to multi-robot path (and motion) 
planning problems.}, 
we are to move multiple rigid bodies (e.g., vehicles, mobile robots, 
quadcopters) from an initial configuration to a goal configuration 
in a collision-free manner. As a key sub-problem in many multi-robot 
applications, fast, optimal, and practical resolution to \mrpp\ are 
actively sought after. On the other hand, \mrpp, in its many forms, 
have been shown to be computationally hard 
\cite{SpiYak84,HopSchSha84,SolHal15,Yu2015IntractabilityPlanar}. 
Due to the high computational complexity of \mrpp, there has 
been a lack of complete algorithms that provide concurrent guarantees 
on computational efficiency and solution optimality. 

\setlength\tabcolsep{3.0pt}
\def\scalefactor{0.75}
\definecolor{colora}{RGB}{244,101,40}
\definecolor{colorb}{RGB}{241,80,149}
\definecolor{colorc}{RGB}{80,168,227}
\definecolor{colord}{RGB}{184,214,4}
\definecolor{colore}{RGB}{127,66,51}
\def\introwidth{3.6}
\def\introheight{3.5}
\def\leftbtmx{0.825}
\def\leftbtmy{-0.2}
\begin{figure}[ht!]
\vspace*{3mm}
    \centering
		\hspace*{-2mm}
    \begin{tabular}{ccccc}
    \begin{tikzpicture}[scale=\scalefactor, every node/.style={scale=\scalefactor}]
        \draw [dashed] (\leftbtmx, \leftbtmy) rectangle (\introwidth + \leftbtmx, \introheight + \leftbtmy);
        \node [robot_node, fill=colora] at (2.95, 1.3) {};
        \node [robot_node, fill=colorb] at (2.35, 1.3) {};
        \node [robot_node, fill=colorc] at (2.65, 1.8) {};
        \node [robot_node, fill=colord] at (2.1, 1.9) {};
        \node [robot_node, fill=colore] at (3.2, 2.0) {};
        \draw [dashed] (2.65, 1.65) circle (0.9);
    \end{tikzpicture}
    &
    \begin{tikzpicture}[scale=\scalefactor, every node/.style={scale=\scalefactor}]
        \draw [dashed] (\leftbtmx, \leftbtmy) rectangle (\introwidth + \leftbtmx, \introheight + \leftbtmy);
        \node [robot_node, fill=colora] at (3.25, 1.0) {};
        \node [robot_node, fill=colorb] at (2.1, 1.0) {};
        \node [robot_node, fill=colorc] at (2.65, 2.2) {};
        \node [robot_node, fill=colord] at (1.8, 2.1) {};
        \node [robot_node, fill=colore] at (3.5, 2.2) {};
    \end{tikzpicture}
    &
    \begin{tikzpicture}[scale=\scalefactor, every node/.style={scale=\scalefactor}]
        \draw [dashed] (\leftbtmx, \leftbtmy) rectangle (\introwidth + \leftbtmx, \introheight + \leftbtmy);
        \foreach \j in {0,...,2}{%
            \ifodd\j 
                \foreach \i in {0,...,2}{
                    \node[hexagon] at ({\j*3/2*\hexagonsize+50},{(2*\i-1)*sin(60)*\hexagonsize+15}) {};}        
            \else
                \foreach \i in {0,...,1}{
                    \node[hexagon] at ({\j*3/2*\hexagonsize+50},{\i*sin(60)*2*\hexagonsize+15}) {};}
            \fi
        \node [robot_node, fill=colord] at (2.1, 2.55) {};
        \node [robot_node, fill=colore] at (3.8, 2.55) {};
        \node [robot_node, fill=colora] at (3.0, 1.05) {};
        \node [robot_node, fill=colorb] at (2.3, 1.05) {};
        \node [robot_node, fill=colorc] at (3.0, 2.0) {};
        }
    \end{tikzpicture}
    \\
    {\footnotesize (a) Initial configuration} & 
    {\footnotesize (b) Expand} & 
    {\footnotesize (c) Assign}
    \vspace*{1.5mm} \\
    \begin{tikzpicture}[scale=\scalefactor, every node/.style={scale=\scalefactor}]
        \draw [dashed] (\leftbtmx, \leftbtmy) rectangle (\introwidth + \leftbtmx, \introheight + \leftbtmy);
        \foreach \j in {0,...,2}{%
            \ifodd\j 
                \foreach \i in {0,...,2}{
                    \node[hexagon] at ({\j*3/2*\hexagonsize+50},{(2*\i-1)*sin(60)*\hexagonsize+15}) {};}        
            \else
                \foreach \i in {0,...,1}{
                    \node[hexagon] at ({\j*3/2*\hexagonsize+50},{\i*sin(60)*2*\hexagonsize+15}) {};}
            \fi
        \node [robot_node, fill=colora] at (3.2, 2.5) {};
        \node [robot_node, fill=colorb] at (2.1, 1.55) {};
        \node [robot_node, fill=colorc] at (3.0, 1.05) {};
        \node [robot_node, fill=colord] at (2.1, 0.55) {};
        \node [robot_node, fill=colore] at (2.1, 2.55) {};
        }
    \end{tikzpicture}
    &
    \begin{tikzpicture}[scale=\scalefactor, every node/.style={scale=\scalefactor}]
        \draw [dashed] (\leftbtmx, \leftbtmy) rectangle (\introwidth + \leftbtmx, \introheight + \leftbtmy);
        \node [robot_node, fill=colora] at (3.3, 2.2) {};
        \node [robot_node, fill=colorb] at (1.9, 1.7) {};
        \node [robot_node, fill=colorc] at (2.9, 1.05) {};   
        \node [robot_node, fill=colord] at (2.0, 0.9) {};
        \node [robot_node, fill=colore] at (2.3, 2.5) {};
    \end{tikzpicture}
    &
    \begin{tikzpicture}[scale=\scalefactor, every node/.style={scale=\scalefactor}]
        \draw [dashed] (\leftbtmx, \leftbtmy) rectangle (\introwidth + \leftbtmx, \introheight + \leftbtmy);
        \node [robot_node, fill=colora] at (3.0, 2.0) {};
        \node [robot_node, fill=colorb] at (2.3, 1.7) {};
        \node [robot_node, fill=colorc] at (2.8, 1.3) {};
        \node [robot_node, fill=colord] at (2.25, 1.15) {};
        \node [robot_node, fill=colore] at (2.5, 2.25) {};
        \draw [dashed] (2.65, 1.65) circle (0.9);
    \end{tikzpicture}
    \\
    {\footnotesize (d) Route} & 
    {\footnotesize (e) Contract} & 
    {\footnotesize (f) Goal configuration}\\
    \end{tabular}
    \caption{\label{fig:intro} 
    Illustration of the expand-assign-route part of the (\esrc) pipeline.
    (a) the initial configuration. (a)$\to$(b) expansion to create space 
		for moving the robots around. (b)$\to$(c) matching and snapping the robots onto 
		a grid structure. (c)$\to$(d) routing the robots on graphs. (d)$\to$(e)$\to$(f) 
		the contraction step, which is the reverse of the expand-assign steps. (f) 
		the goal configuration. 
    }
\end{figure}
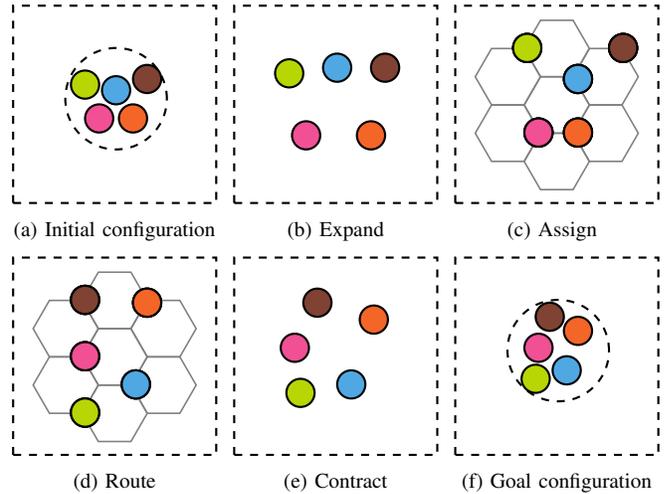

In this paper, we take initial steps to fill this long-standing 
optimality-efficiency gap in labeled multi-robot path planning in 
{\em continuous domains}, without the presence of obstacles. 
We denote this version of the \mrpp problem as \cmpp. 
Through the shift-expand-assign-route (\esrc) solution pipeline and 
the adaptation of a proper discrete robot routing algorithm 
\cite{yu2017expected}, we derive a first polynomial time algorithm 
for solving \cmpp in 2 or 3 dimensions that guarantees constant-factor 
solution optimality on both makespan (i.e., time to complete the 
reconfiguration task) and total distance (i.e., the sum of distances
traveled by all robots), in expectation. 
For both initial and goal configurations, the robots may be located
arbitrarily close to each other. 

This work brings forth two main contributions. First, we introduce 
a novel pipeline, \esrc, for solving \cmpp in an obstacle-free setting
with very high robot density. With a proper discrete routing algorithm, 
we show that \esrc comes with strong theoretical guarantees on both 
computational time and solution optimality, which significantly improves 
over the state-of-the-art, e.g.,~\cite{tang2018complete}. 
Secondly, we have developed practical implementations 
of multiple \esrc-based 
algorithms, including highly efficient (but non-polynomial time) and 
near-optimal ones. As demonstrated through extensive simulation studies, 
\esrc can readily tackle problem instances with a thousand densely 
packed robots. We further demonstrate that \esrc can be applied to 
real multi-robot systems with differential constraints. 

\vspace*{1mm}
\noindent\textbf{Related work}. 
Variations of the multi-robot path planning problem have 
been actively studied for decades 
\cite{ErdLoz86,LavHut98b,GuoPar02,
blm-rvo,StaKor11,SolHal12,TurMicKum14,SolYu15,
wagner2015subdimensional,cohen2016improved,araki2017multi}. 
\mrpp finds applications 
in a diverse array of areas including assembly \cite{HalLatWil00}, evacuation 
\cite{RodAma10}, formation \cite{PodSuk04,SmiEgeHow08}, 
localization \cite{FoxBurKruThr00}, microdroplet manipulation
\cite{GriAke05}, object transportation \cite{RusDonJen95}, search and 
rescue \cite{JenWheEva97}, and human robot interaction \cite{knepper2012pedestrian}, 
to list a few. 
Methods for resolving \mrpp can be {\em centralized}, where a 
central planner dictates the motion of all moving robots, or {\em decentralized}, 
where the moving robots make individual decisions 
\cite{blm-rvo,snape2011hybrid,kim2015velocity}. Given the focus of 
the paper, our review of related literature focuses on centralized methods. 

In a centralized planner, with full access to the system 
state, global planning and control can be readily enforced to drive 
operational efficiency. Methods in this domain may be further classified 
as {\em coupled}, {\em decoupled}, or a mixture of the two. {\em Coupled 
methods} treat all robots as a {\em composite} robot residing in $\mathcal C
=\mathcal C^1 \times \ldots \times \mathcal C^n$, where $\mathcal C^i$ is the 
configuration space of a single robot \cite{schwartz1983piano2}
and subsequently the problem is subjected to standard single robot planning methods 
\cite{latombe2012robot,Lav06}, such as A* in discrete domains
\cite{HarNilRap68} and sampling-based methods in continuous 
domains \cite{SolSalHal14}. However, similar to high dimensional single 
robot problems \cite{reif1985complexity}, \mrpp is strongly NP-hard even for discs 
in simple polygons \cite{SpiYak84}. The hardness of the problem extends to the unlabeled 
case \cite{SolHal15} although complete and optimal algorithms for practical 
scenarios have been developed~\cite{YuLav12CDC,TurMicKum14,adler2015efficient,SolYu15}.

In contrast, {\em decoupled methods} first compute
a path for each individual robot while ignoring all other robots. Interactions 
between paths are only considered {\em a posteriori}. The delayed coupling 
allows more robots to be simultaneously considered but may induce the
loss of completeness \cite{sanchez2002using}. One can however achieve
completeness using optimal decoupling \cite{BerSnoLinMan09}. The classical
approach in decoupling is through {\em prioritization}, in which an order is 
forced upon the robots \cite{ErdLoz86}, potentially significantly affect 
the solution quality\cite{BerOve05}, though remedy is possible
to some extent \cite{bennewitz2002finding}. Sometimes the priority may 
be imposed on the fly, which involves plan-merging and deadlock 
detection \cite{saha2006multi}. 
Often, priority is used in conjunction with {\em velocity tuning}, which 
iteratively decides the velocity of lower priority robots \cite{KanZuc86}.

In the past decade, significant progress has been made on {\em optimally} 
solving (labeled) \mrpp\ problems in discrete settings, in particular on 
discrete (graph-base) environments. Here, the feasibility problem 
is solvable in $O(|V|^3)$ time, in which $|V|$ is the 
number of vertices of the graph on which the robots reside 
\cite{AulMonParPer99,GorHas10,YuArxiv-1301-2342}. Optimal \mrpp\ 
remains computationally intractable in a graph-theoretic setting 
\cite{Gol84,Yu2015IntractabilityPlanar}, but the complexity has dropped 
from PSPACE-hard to NP-complete in many cases. On the algorithmic 
side, decoupling-based heuristics continue to prove to be useful~\cite{StaKor11,wagner2015subdimensional,boyarski2015icbs}. 
Beyond decoupling, other ideas have 
also been explored, often through reduction to other problems 
\cite{Sur12,erdem2013general,YuLav16TRO}.

\vspace*{1mm}
\noindent\textbf{Organization}. The rest of this paper is structured as follows.
In Section~\ref{sec:problem}, we state the problem setting.
In Section~\ref{sec:bound}, after establishing the lower bound on achievable solution optimality, 
we introduce the \esrc\ at a high level.
In Section~\ref{sec:algorithm}, 
we describe the \esrc\ framework in full detail and provide a theoretical 
analysis of its key properties.
In Section~\ref{sec:evaluation}, 
we evaluate the performance of algorithms explained in this paper.
We conclude the paper in Section~\ref{sec:conclusion}.
\section{Problem Formulation}\label{sec:problem}
Consider an obstacle-free $k$ dimensional environment 
($k = 2$ or $k = 3$)
and a set of $n$ labeled spherical (discs for 2D) robots with uniform radius $r$.
The robots move with speed $\|v\|\in [0, 1]$. For a robot $i$,
let $x_i \in \rspace^k$ be its center; a joint configuration 
for all robot is then $X = (x_1, \dots, x_n) \in \rspace^{kn}$.
Let $\|x_i - x_j\|$ be the Euclidean distance between $x_i$ and $x_j$, $1 \le i < j \le n$,
collision avoidance requires $\|x_i - x_j\| > 2r$. 
Suppose the initial and goal configurations of the robots are $X^I$ and $X^G$, 
respectively, so that each robot $i$ is assigned an initial configuration
$x_i^I$ and a goal configuration $x_i^G$. 
Denoting the set of all non-negative real numbers as $\nonneg$, 
a {\em plan} is a function $\mathcal P: \nonneg \to \rspace^{kn}, t \mapsto X^t$
where for all $1 \le i \le n, t, \varepsilon \in \nonneg$, 
$\|x_i^t - x_i^{t + \varepsilon}\| \leq \varepsilon$. 
The continuous \mrpp\ instance 
studied in this paper is defined as follows.

\begin{problem}
    {\normalfont \bf Continuous Multi-Robot Path Planning (\cmpp)}.
    Given $\starts, \goals$, 
    find a plan $\mathcal{P}$ with $\mathcal{P}(0) = X^I$, 
    and $\mathcal{P}(t) = X^G$ for all $t \geq T$.
\end{problem}

We would like to find solutions with 
optimality assurance, according to the following objectives:
1) min-makespan: minimize the required time to finish the task, i.e. $T$;
2) min-total-distance: minimize the cumulative distance traveled by all robots.
In particular, we are interested in situations where $X^I$ and $X^G$ are very dense. 
Therefore, we assume both $X^I$ and $X^G$ are bounded by orthotopes (a 
term generalizing and encompassing rectangles and cuboids) 
with volume linear to the total volume of robots,
i.e., the volume of bounding orthotopes is $\Theta(n r^k)$.

The probability of a particular $X^I$ (and $X^G$) occurring is 
assigned as follows. First, a collision free {\em unlabeled} configuration 
is fixed, for which a probability is difficult to assign. Instead, we 
assume that any labeling of the fixed unlabeled configuration has equal 
probability, i.e., each $X^I$ (and $X^G$) has a probability of $1/n!$, 
which allows us to avoid arguing about continuous probability spaces. 

\begin{figure}[htp]
    \centering
    \begin{tikzpicture}[scale=1]
        \draw [dashed] (0,0) circle (0.9);
        \draw [dashed] (4,0) circle (1);
        \node [point_node] (oi) at (0,0) {};
        \node [point_node] (og) at (4,0) {};
        \node at (-0.8,0.9) {$c^I$};
        \node at (3.2,0.9) {$c^G$};
        \node at (0.2,-0.2) {$o^I$};
        \node at (3.8,0.2) {$o^G$};
        \draw [-, dashed] (oi) to (og);
        \node at (2,-0.2) {$d$};
        \foreach \nodeName/\nodeLocation in {
            1/{(4.5, -0.5)}, 2/{(4.6, 0.4)}, 3/{(3.32, -0.25)}}{
            \node [robot_node, fill=robot_red] 
            (\nodeName) at \nodeLocation {$x_{\nodeName}^G$};
        }
        \foreach \nodeName/\nodeLocation in {
            1/{(0.25, 0.6)}, 2/{(-0.64, 0)}, 3/{(-0.2, -0.6)}}{
			\node [robot_node] (\nodeName) at \nodeLocation {$x_{\nodeName}^I$};
        }
	\end{tikzpicture}
    \caption{\label{fig:problem} 
    An illustration of \cmpp\ when $n = 3$ and $k = 2$.
    The dashed circles are the bounding circles $c^I$ and $c^G$ with radii 
	$r^I$ and $r^G$, respectively. In this case, $r_b = r^G$.
    The blue and red discs show the initial and goal configurations of the robots.
    }
\end{figure}
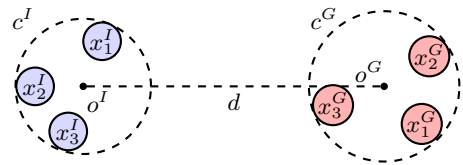

In the following sections, most illustrations and explanations assume $k = 2$.
Our results and proofs work for $k = 3$. 
For developing the results, we introduce some additional notations. 
Given $X^I$ (resp., $X^G$), let $c^I$ (resp., $c^G$) centered at 
$o^I$ (resp., $o^G$) be the smallest bounding circle that encloses 
all $n$ robots. I.e., let the radius of $c^I$ (resp., $c^G)$ be $r^I$ 
(resp., $r^G$), then for all $1 \le i \le n$, $\| o^I - x_i^I\| \le 
r^I - r$ (resp., $\| o^G - x_i^G\| \le r^G - r$). 
We will use $r_b = \max \{r^I, r^G\}$ and $d = \|o^I - o^G\|$ 
to state our optimality results. An illustrative \cmpp\ instance with
additional notations is given in Fig.~\ref{fig:problem}.


\section{Main Results}\label{sec:bound}

\subsection{Lower Bounds on Optimality}

In this sub-section, 
we establish lower bounds on the achievable optimality for \cmpp.
Since interchanging $X^I$ and $X^G$ does not affect optimality lower bounds, 
without loss of generality we make following assumptions: 
1) $r_b = r^G \geq r^I$;
2) $X^I$ is fixed, while the labeling of $X^G$ is by an arbitrary permutation of $\{1, \ldots, n\}$.

\begin{lemma}\label{lem:single_robot_distance}
For any robot $1 \le i \le n$, the probability that 
$\|x_i^I - x_i^G\| = o(r_b)$ is less than or equal to $\frac{1}{4}$.
\end{lemma} 

\begin{proof}
    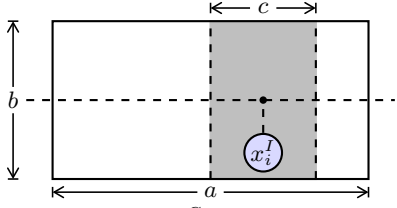
\begin{figure}[htp]
		\vspace*{-2mm}
        \centering
        \begin{tikzpicture}[scale=0.7]
            \draw [fill=gray!50!, gray!50!] (3, 3) rectangle (5, 0); 
            \draw (0, 3) rectangle (6, 0);
            \node [point_node] (xP) at (4, 1.5) {};
            \node [robot_node] (xI) at (4, 0.5) {$x_i^I$};
            \draw [-, dashed] (xI) to (xP);
            \draw [-, dashed] (3, 0) to (3, 3);
            \draw [-, dashed] (5, 0) to (5, 3);
            \draw [-, dashed] (-0.5, 1.5) to (6.5, 1.5);
            \node (a) at (3, -0.25) {$a$};
            \draw [semithick, -] (0, -0.15) to (0, -0.35);
            \draw [semithick, -] (6, -0.15) to (6, -0.35);
            \draw [semithick, ->, >=angle 60] (2.75, -0.25) to (0, -0.25);
            \draw [semithick, ->, >=angle 60] (3.25, -0.25) to (6, -0.25);
            \node (b) at (-0.75, 1.5) {$b$};
            \draw [semithick, -] (-0.85, 0) to (-0.65, 0);
            \draw [semithick, -] (-0.85, 3) to (-0.65, 3);
            \draw [semithick, ->, >=angle 60] (-0.75, 1.75) to (-0.75, 3);
            \draw [semithick, ->, >=angle 60] (-0.75, 1.25) to (-0.75, 0);
            \node (c) at (4, 3.25) {$c$};
            \draw [semithick, -] (3, 3.15) to (3, 3.35);
            \draw [semithick, -] (5, 3.15) to (5, 3.35);
            \draw [semithick, ->, >=angle 60] (3.75, 3.25) to (3, 3.25);
            \draw [semithick, ->, >=angle 60] (4.25, 3.25) to (5, 3.25);
        \end{tikzpicture}
			\vspace*{-2mm}
				\caption{Bounding rectangle of $X^G$ and the shaded area defined by $x_i^I$.}
        \label{fig:bounding-rectangles}
    \end{figure}

    We will demonstrate the number of goal configurations with distance $o(r_b)$ 
    to $x_i^I$ is less than or equal to $\frac{n}{4}$.

    Recall the assumption that $X^I$ and $X^G$ are bounded by orthotopes. 
    As shown in Fig.~\ref{fig:bounding-rectangles}, when $k = 2$, 
    the side lengths of the bounding rectangle of $X^G$ is denoted as $a$ and $b$.
    Without loss of generality, we assume $a \geq b$.
    The assumptions we made indicate that $a = \Theta(r^G) = \Theta(r_b)$, and $ab = \Theta(n r^2)$. 
    Now, consider a rectangle (the shaded area in Fig.~\ref{fig:bounding-rectangles}) 
    centered at the projection of $x_i^I$ on the longer central axis of rectangle $ab$, 
    with side lengths $b$ and $c = \frac{n \pi r^2}{4b}$.
    Since $bc = \frac{n \pi r^2}{4} = \Theta(n r^2)$, we deduce that $c = \Theta(a) = \Theta(r_b)$,
    and the number of goal configurations fully contained 
    in rectangle $bc$ is at most $\frac{n \pi r^2}{4} / \pi r^2 = \frac{n}{4}$.
    Therefore, for an arbitrary labeling in $X^G$, 
    the probability of $x_i^G$ to be fully contained in rectangle $bc$ is at most $\frac{1}{4}$, 
    which covers all the cases that $\|x_i^I - x_i^G\| \leq \frac{c}{2} - r = \Theta(r_b) - r$.
    Since $\pi r_b^2 \geq n \pi r^2$, $\Theta(r_b) - r = \Theta(r_b)$. 
    Thus, the probability that $\|x_i^I - x_i^G\| = o(r_b)$ is less than or equal to $\frac{1}{4}$.
    
    A similar  approach can be applied to the case when $k = 3$. 
    The dimension of the bounding orthotope of $X^G$ is then denoted as 
    $a \times b \times d$ ($a \geq b$, $a \geq d$), 
    and the area around $x_i^I$ is defined as $c \times b \times d$, 
    where $c = \frac{3}{16} \frac{n \pi r^3}{bd}$.
    Details are omitted due to the lack of space.
\end{proof}

We can then reason about optimality lower bounds of \cmpp\ 
based on the $n!$ problem instances 
generated by all the permutations of labeling in $X^G$.

\begin{lemma}\label{lem:lb_makespan}
The fraction of problem instances with $o(r_b)$ optimal makespan is at most $(\frac{1}{2})^{\frac{n}{2}}$.
\end{lemma} 

\begin{proof}
Note that an $o(r_b)$ min-makespan requires $\|x_i^I - x_i^G\| = o(r_b)$ to hold for all $1 \leq i \leq n$.
By Lemma~\ref{lem:single_robot_distance}, for robot $1$, the probability that 
$\|x_1^I - x_1^G\| = o(r_b)$ is at most $\frac{1}{4}$. 
Then, given that the assignments of goal configurations of different robots may be dependent, 
the probability that $\|x_2^I - x_2^G\| = o(r_b)$ is at most $\frac{n/4}{n - 1}$.
Reasoning over the first $\frac{n}{2}$ robots, the probability that $\|x_i^I - x_i^G\| = o(r_b)$ holds 
for all $1 \leq i \leq \frac{n}{2}$ is at most $(\frac{1}{2})^{\frac{n}{2}}$.
\end{proof}

\begin{lemma}\label{lem:lb_total_distance}
The average optimal total travel distance of all the $n!$ instances is $\Omega(n r_b)$.
\end{lemma} 

\begin{proof}
Denote $x_{ij}^G$ as the goal configuration of robot $j$ in problem instance $i$,
the sum of optimal total travel distances of all the $n!$ instances is at least
$\sum_{i = 1}^{n!} \sum_{j = 1}^{n} \|x_j^I - x_{ij}^G\| 
= \sum_{j = 1}^{n} \sum_{i = 1}^{n!} \|x_j^I - x_{ij}^G\| 
\geq \frac{3}{4} n! \sum_{j = 1}^{n} \Theta(r_b)
= n! \, \Theta(n r_b)$.
Thus, the average optimal total travel distance is $\Omega(n r_b)$.
\end{proof}

Taking $d$ into consideration, the optimality lower bounds on \cmpp 
are fully established by the following lemma. 

\begin{lemma}\label{lem:lower_bound}
For obstacle-free \cmpp\ with an arbitrary permutation of robot labeling in $X^I$ and $X^G$, 
the lower bound of makespan is $\Omega(r_b + d)$, in expectation;
the lower bound of total travel distance is $\Omega(n(r_b + d))$,
in expectation.
\end{lemma} 

\begin{proof}
This Lemma extends the results from Lemma~\ref{lem:lb_makespan} 
and Lemma~\ref{lem:lb_total_distance} by covering the case when $d = \omega(r_b)$: 
each robot must travel a distance of at least $d - 2r_b = \Theta(d)$.
\end{proof}

\subsection{\esrc\ and Upper Bounds on Optimality}

The \esrc\ pipeline, as partially illustrated in Fig.~\ref{fig:intro}, 
is outlined in Alg.~\ref{alg:outline}. 
It contains four essential steps: 
{\em (i)} shift, {\em (ii)} expand, {\em (iii)} assign, and {\em (iv)} route. 
At the start of the pseudo code (line~\ref{alg:outline_translate}), 
the {\em shift} step moves the robots  
so that the centers of the bounding circles $c^I$ and $c^G$ coincide.
Then, the {\em expand} and {\em assign} steps in 
lines~\ref{alg:outline_expand}-\ref{alg:outline_snap} do the 
following: {\em (i)} expand the robots so that they are sufficiently 
apart, {\em (ii)} generate an underline grid $G(V, E)$, and {\em (iii)} 
perform an injective assignment from robot labels to $V$. This yields 
a discrete \mrpp\ sub-problem. Line~\ref{alg:outline_route} 
solves the new \mrpp\ problem. Line~\ref{alg:outline_contract} 
then reverse the expand-assign to reach $X^G$. Note that the sub-plan
for the {\em contract} step is generated together with the 
expand and assign steps. In particular, only a single grid $G$ is 
created. Because of this, we do not include contract as a unique step of 
the \esrc pipeline.

\begin{algorithm}
\begin{small}
	\DontPrintSemicolon
    {\bf Shift}: move all robots so that $o^I$ coincides with $o^G$\;\label{alg:outline_translate}
    {\bf Expand}: increase clearance between the robots\;\label{alg:outline_expand}
    Construct a grid graph $G(V, E)$\;\label{alg:outline_graph}
    {\bf Assign}: move robots based on labeling $\{1, \ldots, n\} \xhookrightarrow{} V$\;\label{alg:outline_snap}
    {\bf Route}: solve the discrete \mrpp\ on $G$ \;\label{alg:outline_route}
    {\bf Contract}: reversed expand-assign steps\;\label{alg:outline_contract}
	\caption{The \tsm Pipeline}
	\label{alg:outline}
\end{small}
\end{algorithm}

We briefly analyze the optimality guarantee provided by \esrc; the details
will follow in the next section. For an arbitrary robot, the shift 
step brings a distance cost of $d$. The expand step incurs a cost of 
$O(r_b)$ and the assign step a cost of $O(r)$. This suggests 
that the contract step also incurs a cost of $O(r_b + r) = O(r_b)$. The route
step, using the algorithm introduced in \cite{yu2017expected}, 
incurs a cost of $O(r_b)$ as well. 
Adding everything together, the distance traveled by a single robot is $O(r_b + d)$. 
Since the \esrc\ pipeline solve the problem for all
robots simultaneously, the makespan is $O(r_b + d)$ and the total distance 
is $O(n(r_b + d))$. In viewing Lemma~\ref{lem:lower_bound}, we 
summarize the section with the following theorem, to be fully proved in the next section. 

\begin{theorem}\label{t:main}
For obstacle-free \cmpp\ with an arbitrary permutation of 
robot labeling in $X^I$ and $X^G$, a \esrc-based algorithm can compute constant-factor 
optimal solutions for both makespan and total traveled distance, in expectation.
\end{theorem}
 
\begin{remark}
Between the two objectives, makespan and total distance, our guarantee on makespan 
is stronger. For a problem $p_i$, let its minimum possible makespan be $t_{\min}(p_i)$ 
and let the computed solution have makespan $t_a(p_i)$, we in fact establish that the 
quantity $\sum_i\frac{t_a(p_i)}{t_{\min}(p_i)}$ is a constant. 
\end{remark}
\section{The \esrc\ Pipeline in Detail}\label{sec:algorithm}

\subsection{Shifting the Configuration Center}

Lemma~\ref{lem:lower_bound} indicates that solving part 
of the problem with makespan $O(d)$ will not break the constant 
factor optimality guarantee. Therefore, the \esrc pipeline starts with a
shift step that translates the robots
without changing their relative positions, so that $o^I$ is moved to $o^G$
(see Fig.~\ref{fig:translation}).
During the process, the path travelled by robot $i$ is a straight line 
starts at $x_i^I$ and follows vector $\overrightarrow{o^I o^G}$.
Since the workspace is free of obstacles, 
the shift step is always feasible and yields an updated \cmpp with $d= 0$. 
This process incurs exactly $d$ makespan and $nd$ total travel distance.
The computational complexity is $O(n)$.

\begin{figure}[htp]
    \centering
    \begin{tikzpicture}[scale=1]
        \draw [dashed] (0,0) circle (0.9);
        \draw [dashed] (4,0) circle (1);
        \foreach \nodeName/\nodeLocation in {
            1/{(0.25, 0.6)}, 2/{(-0.64, 0)}, 3/{(-0.2, -0.6)}}{
			\node [robot_node] (\nodeName) at \nodeLocation {$x_{\nodeName}^I$};
        }
        \foreach \nodeName/\nodeLocation in {
            4/{(4.25, 0.6)}, 5/{(3.4, 0)}, 6/{(3.8, -0.6)}}{
            \node [robot_node, dashed, fill=robot_green] 
            (\nodeName) at \nodeLocation {};
        }
		\foreach \edgeFrom/\edgeTo in {1/4, 2/5, 3/6}{
			\draw [->] (\edgeFrom) to (\edgeTo);
        }
        \foreach \nodeName/\nodeLocation in {
            1/{(4.5, -0.5)}, 2/{(4.6, 0.4)}, 3/{(3.32, -0.25)}}{
            \node [robot_node, fill=robot_red] 
            (\nodeName) at \nodeLocation {$x_{\nodeName}^G$};
        }
        \node at (-0.8,0.9) {$c^I$};
        \node at (3.2,0.9) {$c^G$};
	\end{tikzpicture}
    \caption{\label{fig:translation}
        An illustration of shift step for the scenario in Fig.~\ref{fig:problem}.}
\end{figure}
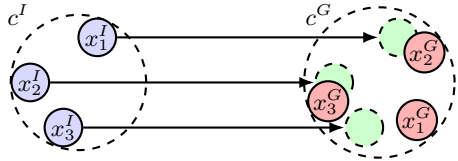

Without loss of generality, in the remaining part of this section, 
we assume $o^I = o^G$, and $d = 0$.

\subsection{Expand and Assign}
We then apply {\em expand} and {\em assign} steps to both $X^I$ and $X^G$. 
The objective is to convert \cmpp to a \mrpp on a grid graph $G = (V, E)$.
The edge length of $G$ is dependent on $r$, the radius of robots. We 
illustrate the steps using $X^I$; the same procedure is applied to $X^G$ in 
the reverse order. 

\subsubsection{Expand}
spread the robots
so that there is enough clearance between them to move them onto $G$. 
For doing so, we use simple {\em linear expansion} as defined below.

\begin{definition} {\bf Linear Expansion.}
    Denote $\lambda \geq 1$ as a {\em scaling factor}, 
    each robot $i$ then moves along the direction 
    of vector $\overrightarrow{o^I x_i^I}$ at maximum speed $1$,
    travels a distance $(\lambda - 1) \|o^I - x_i^I\|$.
\end{definition}

The value of $\lambda$, to be decided later, depends on grid property
and also the placement of robots in $X^I$. 
Now, we first show that linear expansion is collision free. 

\begin{lemma}
    Linear expansion is always collision-free.
\end{lemma}
\begin{proof}
    Because linear expansion with scaling factor $\lambda$ 
    applies different travel distances to robots,
    we need to eliminate potential collisions 
    between two moving robots,
    as well as between one moving robot and another stopped robot.

    A pair of moving robots $i$ and $j$ 
    always travel the same distance $t$ during time $[0, t]$
    because they both move at the maximum speed $1$.
    As shown in Fig.~\ref{fig:linear_expansion}(a),
    we denote $a = \|o^I - x_i^I\|, b = \|o^I - x_j^I\|$, 
    and $\alpha = \angle~x_i^I o^I x_j^I$.
    The increased distance between $i$, $j$ during linear expansion is 
    \begin{equation*}
        \begin{aligned}
                & \, (a + t)^2 + (b + t)^2 - 2(a + t)(b + t)\cos\alpha - \|x_i^I - x_j^I\| \\
            =   & \, 2(1 - \cos\alpha)(t a + t b + t^2) \geq 0.
        \end{aligned}
    \end{equation*}
    Since $\|x_i^I - x_j^I\| > 2r$,
    two moving robots cannot collide.

    The robots which are stopped earlier 
    will not collide with any other robots afterwards.
    See Fig.~\ref{fig:linear_expansion}(b),
    after robot $1$ stopped moving, the robots closer to $o^I$ 
    (robot $2$, $3$, in the shaded area) are already stopped.
    As other robots (robot $4$, $5$, $6$, centered out of the shaded area) 
    move away from the shaded area, 
    the distance between any of them and robot $1$ monotonically increases.
    So they will not collide with robot $1$ in the future.
\end{proof}

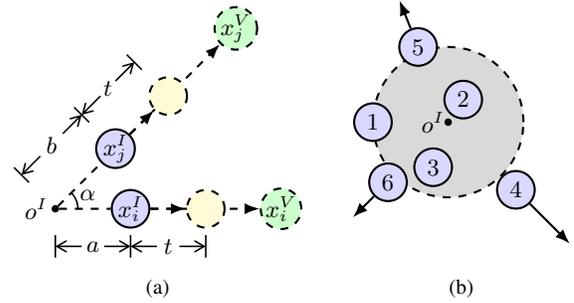
\begin{figure}[htp]
    \centering
    \begin{tabular}{ccc}
    \begin{tikzpicture}[scale=1]
        \node at (-0.25,0) {$o^I$};
        \draw [semithick] (0, 0) ++(0:.3) arc (0:45:.3);
        \node at (0.43, 0.18) {$\alpha$};
        \node [point_node] (ci) at (0, 0) {};
        \node [robot_node, dashed, fill=robot_green] (v1) at (3, 0) {$x_i^V$};
        \node [robot_node, dashed, fill=robot_green] (v2) at (2.4, 2.4) {$x_j^V$};
        \node [robot_node] (0) at (1, 0) {$x_i^I$};
        \node [robot_node] (1) at (0.8, 0.8) {$x_j^I$};
		\foreach \edgeFrom/\edgeTo in {ci/0, ci/1}{
			\draw [-, dashed] (\edgeFrom) to (\edgeTo);
        }
		\foreach \edgeFrom/\edgeTo in {0/v1, 1/v2}{
			\draw [->, dashed] (\edgeFrom) to (\edgeTo);
        }
        \foreach \nodeName/\nodeLocation in {
            2/{(2, 0)}, 3/{(1.5, 1.5)}}{
            \node [robot_node, dashed, fill=robot_yellow]
            (\nodeName) at \nodeLocation {};
        }
		\foreach \edgeFrom/\edgeTo in {0/2, 1/3}{
			\draw [->, dashed] (\edgeFrom) to (\edgeTo);
        }
        \node (a) at (0.5, -0.5) {$a$};
        \node (da) at (1.5, -0.5) {$t$};
        \draw [semithick, ->, >=angle 60] (0.35, -0.5) to (0, -0.5);
        \draw [semithick, ->, >=angle 60] (0.65, -0.5) to (1, -0.5);
        \draw [semithick, ->, >=angle 60] (1.35, -0.5) to (1, -0.5);
        \draw [semithick, ->, >=angle 60] (1.65, -0.5) to (2, -0.5);
        \draw [semithick, -] (0, -0.35) to (0, -0.65);
        \draw [semithick, -] (1, -0.35) to (1, -0.65);
        \draw [semithick, -] (2, -0.35) to (2, -0.65);
        \draw [semithick, -] (-0.35, 0.35) to (-0.55, 0.55);
        \draw [semithick, -] (0.8-0.35, 0.8+0.35) to (0.8-0.55, 0.8+0.55);
        \draw [semithick, -] (1.5-0.35, 1.5+0.35) to (1.5-0.55, 1.5+0.55);
        \node (b) at (0.4-0.45, 0.4+0.45) {$b$};
        \node (da) at (1.15-0.45, 1.15+0.45) {$t$};
        \draw [semithick, ->, >=angle 60] 
        (0.55-0.45, 0.55+0.45) to (0.8-0.45, 0.8+0.45);
        \draw [semithick, ->, >=angle 60] 
        (0.25-0.45, 0.25+0.45) to (-0.45, 0.45);
        \draw [semithick, ->, >=angle 60] 
        (1.25-0.45, 1.25+0.45) to (1.5-0.45, 1.5+0.45);
        \draw [semithick, ->, >=angle 60] 
        (1.05-0.45, 1.05+0.45) to (0.8-0.45, 0.8+0.45);
	\end{tikzpicture}
        &&
    \begin{tikzpicture}[scale=1]
        \node [robot_node, minimum size=2cm, fill=gray!30!, dashed]
        (ci) at (0, 0) {};
        \foreach \nodeName/\nodeLocation in {
            1/{(-1, 0)}, 2/{(0.2, 0.3)}, 3/{(-0.2, -0.6)}, 4/{(0.9, -0.9)},
5/{(-0.4, 1)}, 6/{(-0.8, -0.8)}}{
			\node [robot_node] (\nodeName) at \nodeLocation {${\nodeName}$};
        }
        \foreach \nodeName/\nodeLocation in {
            7/{(1.75, -1.75)}, 8/{(-0.7, 1.75)}, 9/{(-1.4, -1.4)}}{
			\node (\nodeName) at \nodeLocation {};
        }
		\foreach \edgeFrom/\edgeTo in {4/7, 5/8, 6/9}{
			\draw [->] (\edgeFrom) to (\edgeTo);
        }
        \node at (-0.2,0) {$o^I$};
        \node [point_node] (oi) at (0,0) {};
	\end{tikzpicture}
        \\
        {\footnotesize (a)} && {\footnotesize (b)} \\
    \end{tabular}
    \caption{\label{fig:linear_expansion} Linear expansion.
            (a) Any two moving robots are collision-free.
            (b) The robots stopped will not collide with any other robot afterwards.
    }
\end{figure}

\begin{corollary}\label{lem:clearance}
    The distance between the center of any two robots
    after linear expansion is scaled up by a factor of $\lambda$.
\end{corollary}
\begin{proof}
    See Fig.~\ref{fig:linear_expansion}(a), 
    assume $x_i^V, x_j^V$ are the configurations of 
    robot $i, j$ after linear expansion, respectively.
    Because $\lambda = \|o^I - x_i^V\| / \|o^I - x_i^I\| = \|o^I - x_j^V\| / \|o^I - x_j^I\|$,
    $\lambda = \|x_i^V - x_j^V\| / \|x_i^I - x_j^I\|$.
\end{proof}

During linear expansion, 
because the travel distance of a robot is at most $(\lambda - 1) (r_b - r)$, 
the makespan is $O(\lambda r_b)$ and the total travel distance is $O(n \lambda r_b)$.

\begin{remark}
    Depending on the specific problem instance, expansion could 
    be achieved through multiple ways, as long as the robots move in a 
    collision-free manner. That is, expand can be done in a smarter way.
    However (intuitively and as we will show), linear expansion with a 
    small constant $\lambda$ is sufficient for the \esrc analysis 
    to carry through.
\end{remark}

\subsubsection{Assign}

After expansion, we generate a regular grid graph $G$ 
in an area which covers all robots.
Here, we use hexagonal grid for illustration due to its capability 
for accommodating high robot density and natural support for certain
differential constraints~\cite{yu2018effective}. 
In this case, to ensure the robots can move along edges in $G$ 
without collisions, the edge/side length $\ell$ should be greater than 
$(4/\sqrt{3}) r$~\cite{yu2018effective}. We first establish the 
number of vertices in $G$ so that $G$ fully covers the expanded $X^I$ and $X^G$. 

\begin{lemma}\label{lem:num_vertices}
    The number of vertices in $G$ is $O(\lambda^k r_b^k / \ell^k)$.
\end{lemma}
\begin{proof}
    After linear expansion, the side length of 
    the smallest bounding orthotope of all robots is not larger than $\lambda r_b$.
    If we construct $G$ in this area, then $|V| = O(\lambda^k r_b^k / \ell^k)$.
\end{proof}

The next step is to place the robots on $V$.
This is done by {\em assign} each robot to the nearest vertex,
and execute straight line paths at the full speed $1$.
For a given $X^I$, we define
$d_{min} = \min_{1\le i < j \le n} \|x_i^I - x_j^I\|.$    

\begin{lemma}
    When $\lambda \geq 2\ell / d_{min}$, the following statements holds:
    (i) the mapping between robot labels $\{1, \ldots, n\}$ and $V$ is injective;  
    (ii) the assigned paths are collision-free.
\end{lemma}
\begin{proof}
    {\em (i)} When $\lambda \geq 2\ell / d_{min}$, 
    Corollary~\ref{lem:clearance} implies that after linear expansion,
    robots are located in non-overlapping circles 
    with radius $\ell$ (see Fig.~\ref{fig:match}(a)).
    Hence, there must be at least one vertex in $V$ lands 
    inside or on the border of each circle,
    which is assigned to the robot in the circle. 
    This ensures the injective mapping between robot labels and $V$.

    {\em (ii)} As it is shown in Fig~\ref{fig:match}(b), 
    we denote the configurations of robot $i, j$ after expansion as $x_i^V, x_j^V$, 
    and the vertices assigned to the two robots as $v_i$, $v_j$, respectively. 
    From {\em (i)}, we deduce that $\|x_i^V - v_i\| \leq \ell, 
    \|x_i^V - v_j\| \geq \ell$, and $\|x_i^V - x_j^V\| \geq 2\ell$. 
    Since $\ell$ denote edge length in $G$, $\|v_i - v_j\| \geq \ell$.
    These inequations imply that the distance (the red dashed line)
    from $v_j$ to segment $x_i^V v_i$ is at least $(\sqrt{3}/2) \ell$.
    The above statements also holds for $v_i$ and $x_j^V v_j$.  
    Thus, the shortest distance between segments 
    $x_i^V v_i$ and $x_j^V v_j$ is 
    $(\sqrt{3}/2) \ell > (\sqrt{3}/2) (4/\sqrt{3}) r = 2r$.
    No collisions between robot $i, j$.
\end{proof}

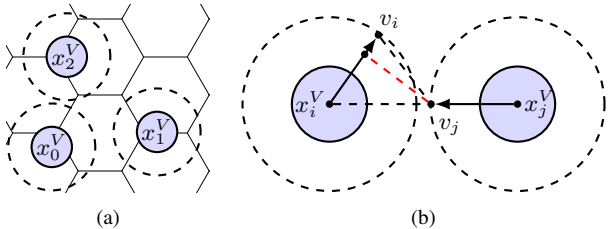
\begin{figure}[htp]
    \centering
    \vspace*{2mm}
    \begin{tabular}{ccc}
    \begin{tikzpicture}[scale=1]
        \fill [pattern=hexagons] (0,0) rectangle (2.7,2.5);
        \foreach \nodeName/\nodeLocation in {
            0/{(0.6, 0.6)}, 1/{(2, 0.8)}, 2/{(0.8, 1.8)}}{
            \node [robot_node] (\nodeName) at \nodeLocation {   $x_{\nodeName}^V$};
            \draw [dashed] \nodeLocation circle (\hexagonsize);
        }        
	\end{tikzpicture}
        &&
    \begin{tikzpicture}[scale=1]
        \node [robot_node, minimum size=1cm] (0) at (0, 0) {};
        \node [robot_node, minimum size=1cm] (1) at (2.5, 0) {};
        \draw [dashed] (0, 0) circle (\hexagonsize * 2);
        \draw [dashed] (2.5, 0) circle (\hexagonsize * 2);
        \node [point_node] (vi) at (0.65, 0.92) {};
        \node [point_node] (vm) at (0.65 * 0.72, 0.92 * 0.72) {};
        \node [point_node] (vj) at (1.35, 0) {};
        \node [point_node] (ri) at (0, 0) {};
        \node [point_node] (rj) at (2.5, 0) {};
        \draw [-, dashed, red] (vj) to (vm);
		\foreach \edgeFrom/\edgeTo in {vi/vj, ri/vj}{
			\draw [-, dashed] (\edgeFrom) to (\edgeTo);
        }
		\foreach \edgeFrom/\edgeTo in {ri/vi, rj/vj}{
			\draw [->] (\edgeFrom) to (\edgeTo);
        }
        \node at (0.8, 1.1) {$v_i$};
        \node at (1.6, -0.3) {$v_j$};
        \node at (-0.25, 0) {$x_i^V$};
        \node at (2.8, 0) {$x_j^V$};
	\end{tikzpicture}
        \\
        {\footnotesize (a)} && {\footnotesize (b)} \\
    \end{tabular}
    \caption{\label{fig:match} Assign.
            (a) The robots are in circles with radius $\ell$ which do not overlap.
            (b) No collision between the trajectories.
    }
\end{figure}

Denote $\varepsilon$ as a small positive real number, 
if we let $\ell = (4/\sqrt{3}) r + \varepsilon$, 
then $2\ell/ d_{min} \leq 2\ell / 2r = 4/ \sqrt{3} + \varepsilon / 2r$.
In this case, it is safe to say that when $\lambda = 4/\sqrt{3} + \varepsilon / 2r$,
the expand-assign can always be successfully performed.
$\lambda$ can be smaller when $d_{min}$ is larger.
The makespan for the whole expand-assign step is $O(r_b)$,
and the total travel distance is $O(n r_b)$.
In viewing Lemma~\ref{lem:num_vertices}, the computational time is $O(n |V|)$,
with the dominating factor being matching the robots with their nearest vertices.

An alternative for the expand-assign step is to 
consider matching robot labels to $V$ as 
an unlabeled \mrpp\ problem in continuous domain.
Because the makespan for expand-assign is already sufficient 
for achieving the desired optimality, 
we opt to keep the method fast and straightforward.

\begin{remark}
In the actual implementation, we also added an extra step that ``compacts''
the \mrpp instance, which makes the initial and goal configurations on 
the graph $G$ more concentrated. The extra step, based largely on 
\cite{YuLav12CDC}, does not impact the asymptotic performance of the 
\esrc pipeline but produces better constants. 
\end{remark}

\subsection{Discrete Robot Routing}

After applied expand and assign steps to both $X^I$ and $X^G$, 
a \cmpp instance becomes a discrete \mrpp instance. Let the converted
instance be $(G, V^I, V^G)$, i.e., $V^I = \{v_1^I,\ldots, v_n^I\}$ and 
$V^G = \{v_1^G,\ldots, v_n^G\}$ are the discrete initial and goal
configurations of the robots on $G$. 
To solve the \mrpp instance with polynomial time complexity and 
constant factor optimality guarantee, we apply the \textsc{SplitAndGroup}~(\sag) 
algorithm~\cite{yu2017expected}. The algorithm is intended for a 
fully occupied graph, but is also directly applicable to sparse 
graphs: we may place ``virtual robots'' on empty vertices. For 
completeness, we provide a brief overview of $\sag$ here. 

\sag\ is a recursive algorithm that splits a \mrpp instance 
into smaller sub-problems in each level of recursive call.
We demonstrate one recursion here as an illustration.
As shown in Fig.~\ref{fig:split}(a), 
{\em split} partitions the part of grid it works on into two areas 
(left side with a blue shade and right side with a red shade) 
along the longer side of this grid.
For each side, the algorithm finds all the robots currently in this side 
with goal vertices in the other side.
Here in Fig.~\ref{fig:split}, 
the blue (resp. red) robots have goal vertices in blue (resp. red) shade.
Then {\em group} moves every robot to the side 
where its goal vertex belongs to (see Fig.~\ref{fig:split}(b)).
This results in two sub-problems (the left one and the right one),
which could be treated by the next level of recursion call in parallel.

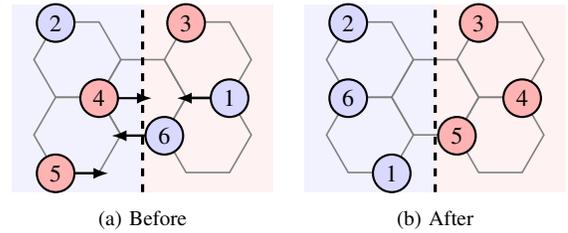
\begin{figure}[htp]
    \centering
    \vspace*{2mm}
    \begin{tabular}{ccc}
    \begin{tikzpicture}[scale=1]
        \fill [blue!5!] 
        (-1.5*\hexagonsize, -0.25) rectangle (\hexagonsize*1.5, 2.25);
        \fill [red!5!] 
        (\hexagonsize*1.5, -0.25) rectangle (\hexagonsize*4.5, 2.25);
        \foreach \j in {0,...,2}{%
            \ifodd\j 
                \foreach \i in {1,...,1}{
                    \node[hexagon] at 
                    ({\j*3/2*\hexagonsize},{(2*\i-1)*sin(60)*\hexagonsize}) {};}        
            \else
                \foreach \i in {0,...,1}{
                    \node[hexagon] at 
                    ({\j*3/2*\hexagonsize},{\i*sin(60)*2*\hexagonsize}) {};}
            \fi
        }
        \draw [-, dashed, very thick] 
        (\hexagonsize*1.5, -0.25) to (\hexagonsize*1.5, 2.25);
        \draw [->] (\hexagonsize*2, 0.5) to (\hexagonsize*2-20, 0.5);
        \draw [->] (\hexagonsize*3.5, 1) to (\hexagonsize*3.5-20, 1);
        \draw [->] (\hexagonsize*0.5, 1) to (\hexagonsize*0.5+20, 1);
        \draw [->] (-\hexagonsize*0.5, 0) to (-\hexagonsize*0.5+20, 0);
        \foreach \nodeName/\nodeLocation in {
            6/{(\hexagonsize*2, 0.5)}, 1/{(\hexagonsize*3.5, 1)}, 2/{(-\hexagonsize*0.5, 2)}}{
            \node [robot_node] (\nodeName) at \nodeLocation {\nodeName};
        }   
        \foreach \nodeName/\nodeLocation in {
            3/{(\hexagonsize*2.5, 2)}, 4/{(\hexagonsize*0.5, 1)}, 5/{(-\hexagonsize*0.5, 0)}}{
            \node [robot_node, fill=robot_red] 
            (\nodeName) at \nodeLocation {\nodeName};
        }      
	\end{tikzpicture}
        &&
    \begin{tikzpicture}[scale=1]
        \fill [blue!5!] 
        (-1.5*\hexagonsize, -0.25) rectangle (\hexagonsize*1.5, 2.25);
        \fill [red!5!] 
        (\hexagonsize*1.5, -0.25) rectangle (\hexagonsize*4.5, 2.25);
        \foreach \j in {0,...,2}{%
            \ifodd\j 
                \foreach \i in {1,...,1}{
                    \node[hexagon] at 
                    ({\j*3/2*\hexagonsize},{(2*\i-1)*sin(60)*\hexagonsize}) {};}        
            \else
                \foreach \i in {0,...,1}{
                    \node[hexagon] at 
                    ({\j*3/2*\hexagonsize},{\i*sin(60)*2*\hexagonsize}) {};}
            \fi
        }
        \draw [-, dashed, very thick] 
        (\hexagonsize*1.5, -0.25) to (\hexagonsize*1.5, 2.25);
        \foreach \nodeName/\nodeLocation in {
            6/{(-\hexagonsize*0.5, 1)}, 1/{(\hexagonsize*0.5, 0)}, 2/{(-\hexagonsize*0.5, 2)}}{
            \node [robot_node] (\nodeName) at \nodeLocation {\nodeName};
        }   
        \foreach \nodeName/\nodeLocation in {
            3/{(\hexagonsize*2.5, 2)}, 4/{(\hexagonsize*3.5, 1)}, 5/{(\hexagonsize*2, 0.5)}}{
            \node [robot_node, fill=robot_red] 
            (\nodeName) at \nodeLocation {\nodeName};
        } 
	\end{tikzpicture}
        \\
        {\footnotesize (a) Before} && {\footnotesize (b) After} \\
    \end{tabular}
    \caption{\label{fig:split} Before and after one recursive call of \sag.}
\end{figure}

The grouping procedure depends on the robots being able to ``swap'' locations
locally in constant time. On a hexagonal grid, swapping any two robots on a fully 
occupied {\em figure-8 graph} (Fig.~\ref{fig:8}) can be performed in 
constant time through an adaptation of the proof used in proving 
Lemma 1 of~\cite{yu2017expected}. As a larger grid is partitioned into 
smaller figure-8 graphs, parallel swaps become possible. 
Let $G$ be bounded in a rectangle of size $m_{w}\times m_{h}$ in which 
$m_w$ and $m_h$ are the width and height of the rectangle, then the 
first iteration of \sag requires a robot incur a makespan of 
$O(m_w + m_h)$ and this can be done in parallel for all robots. 
Recursively, 
the makespan upper bound of \sag\ (applied to our problem) is:
\begin{equation}
\begin{array}{ll}
 & \displaystyle O(m_w + m_h) + O(\frac{m_w}{2} + m_h) + O(\frac{m_w}{2} + \frac{m_h}{2}) + \dots \\
= & \displaystyle O(m_w + m_h) = O(\lambda r_b) = O(r_b). \notag
\end{array}
\end{equation}
The last equality is due to $\lambda$ being a small constant.

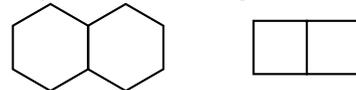
\begin{figure}[htp]
    \centering
    \begin{tabular}{
        >{\centering\arraybackslash}m{2cm}
        >{\centering\arraybackslash}m{0.5cm}
        >{\centering\arraybackslash}m{2cm}}
    \begin{tikzpicture}[scale=1]
        \node[hexagon, rotate=90, draw=black, thick] at (0, 0) {};
        \node[hexagon, rotate=90, draw=black, thick] at (1, 0) {};
    \end{tikzpicture}
    &&
    \begin{tikzpicture}[scale=0.708]
        \draw (0, 0) rectangle (1, 1);
        \draw (2, 0) rectangle (1, 1);
    \end{tikzpicture}
    \\
    \end{tabular}
    \caption{\label{fig:8} Figure-8 graphs in hexagonal grids and square grids.}
\end{figure}

\sag incurs $O(|V|^3) = O((\lambda^k r_b^k / l^k)^3) = O((\lambda^k n r^k / l^k)^3) = O(n^3)$ computational time.

\subsection{Performance Analysis of \tsm}

Summarizing the result from individual steps, for the general \esrc 
pipeline with the shift step, we conclude

\begin{theorem}\label{t:esrc}
    For obstacle-free \cmpp, \esrc computes solutions with $O(r_b + d)$
		makespan and $O(n(r_b + d))$ total distance in $O(n^3)$ computational time. 
\end{theorem}

Theorem~\ref{t:main} follows Lemma~\ref{lem:lower_bound}
and Theorem~\ref{t:esrc}. 

\subsection{Generalization to a 3 Dimensional Workspace}
\esrc is ready to be generalized to 
a 3 dimensional workspace with some modifications to key parameters. 
The shift step can be applied to an arbitrary dimension.
In the expand and assign steps, suppose a cube grid is used, 
the edge length $l$ should be at least $(4 / \sqrt{2})r + \varepsilon$ 
to ensure that robots can move along edges without collisions.
A scale factor $\lambda = 4 / \sqrt{2} + \varepsilon / 2r$ 
is sufficient for assign to be collision-free.
For any fixed dimension, the guarantee of the \sag algorithm 
continuous to hold \cite{yu2017expected}. With these modifications, 
Theorem~\ref{t:esrc} and Theorem~\ref{t:main} holds for $k = 3$.

%

%


\section{Evaluation}\label{sec:evaluation}

\pgfplotsset{every axis/.append style={
    font=\scriptsize,
    /pgf/number format/1000 sep={},
    height = 0.46 \linewidth,
    width = 0.52 \linewidth,
    x label style={at={(axis description cs:0.5,-0.1)}},
    y label style={at={(axis description cs:-0.05,0.5)}},
    ymajorgrids=true,
    grid style=dashed
}}
\tikzset{mark size=1.25}

\pgfplotsset{select coords between index/.style 2 args={
    x filter/.code={
        \ifnum\coordindex<#1\def\pgfmathresult{}\fi
        \ifnum\coordindex>#2\def\pgfmathresult{}\fi
    }
}}

\makeatletter
\def\pgfplots@getautoplotspec into#1{%
    \begingroup
    \let#1=\pgfutil@empty
    \pgfkeysgetvalue{/pgfplots/cycle multi list/@dim}\pgfplots@cycle@dim
    \let\pgfplots@listindex=\pgfplots@numplots
    \pgfkeysgetvalue{/pgfplots/cycle list set}\pgfplots@listindex@set
    \ifx\pgfplots@listindex@set\pgfutil@empty
    \else 
        \c@pgf@counta=\pgfplots@listindex
        \c@pgf@countb=\pgfplots@listindex@set
        \advance\c@pgf@countb by -\c@pgf@counta
        \globaldefs=1\relax
        \edef\setshift{%
            \noexpand\pgfkeys{
                /pgfplots/cycle list shift=\the\c@pgf@countb,
                /pgfplots/cycle list set=
            }
        }%
        \setshift%
    \fi
    \pgfkeysgetvalue{/pgfplots/cycle list shift}\pgfplots@listindex@shift
    \ifx\pgfplots@listindex@shift\pgfutil@empty
    \else
        \c@pgf@counta=\pgfplots@listindex\relax
        \advance\c@pgf@counta by\pgfplots@listindex@shift\relax
        \ifnum\c@pgf@counta<0
            \c@pgf@counta=-\c@pgf@counta
        \fi
        \edef\pgfplots@listindex{\the\c@pgf@counta}%
    \fi
    \ifnum\pgfplots@cycle@dim>0
        %
        \c@pgf@counta=\pgfplots@cycle@dim\relax
        \c@pgf@countb=\pgfplots@listindex\relax
        \advance\c@pgf@counta by-1
        \pgfplotsloop{%
            \ifnum\c@pgf@counta<0
                \pgfplotsloopcontinuefalse
            \else
                \pgfplotsloopcontinuetrue
            \fi
        }{%
            \pgfkeysgetvalue{/pgfplots/cycle multi list/@N\the\c@pgf@counta}\pgfplots@cycle@N
            \pgfplotsmathmodint{\c@pgf@countb}{\pgfplots@cycle@N}%
            \divide\c@pgf@countb by \pgfplots@cycle@N\relax
            \expandafter\pgfplots@getautoplotspec@
                \csname pgfp@cyclist@/pgfplots/cycle multi list/@list\the\c@pgf@counta @\endcsname
                {\pgfplots@cycle@N}%
                {\pgfmathresult}%
            \t@pgfplots@toka=\expandafter{#1,}%
            \t@pgfplots@tokb=\expandafter{\pgfplotsretval}%
            \edef#1{\the\t@pgfplots@toka\the\t@pgfplots@tokb}%
            \advance\c@pgf@counta by-1
        }%
    \else
        \pgfplotslistsize\autoplotspeclist\to\c@pgf@countd
        \pgfplots@getautoplotspec@{\autoplotspeclist}{\c@pgf@countd}{\pgfplots@listindex}%
        \let#1=\pgfplotsretval
    \fi
    \pgfmath@smuggleone#1%
    \endgroup
}

\pgfplotsset{
    cycle list set/.initial=
}
\makeatother

\newenvironment{customlegend}[1][]{%
    \begingroup
    \csname pgfplots@init@cleared@structures\endcsname
    \pgfplotsset{#1}%
}{%
    \csname pgfplots@createlegend\endcsname
    \endgroup
}%

\def\addlegendimage{\csname pgfplots@addlegendimage\endcsname}

\pgfplotstableread{datan.dat}{\datan}
\pgfplotstableread{datae.dat}{\datae}
\pgfplotstableread{datad.dat}{\datad}
\pgfplotstableread{datan2.dat}{\datans}

\setlength\tabcolsep{1.0pt} 
\renewcommand{\arraystretch}{0.3} 

In this section, we provide evaluations of computational time 
and solution optimality of \tsm-based algorithms when $k = 2$. A hardware experiment is also
carried out to verify the support for differential constraints. In our implementation,
beside the baseline \esrc with \sag, we also implemented an \esrc algorithm 
with the high quality (but non-polynomial-time) ILP solver from \cite{YuLav16TRO}.
We denote \tsm\ with \sag\ as \tsmsag, 
and \tsm\ with ILP solver as \tsmilp.
\tsmilp-$m$ indicates the usage of $m$-way split heuristics \cite{YuLav16TRO}.
We use makespan as the key metric which 
also bounds the total distance and maximum distance.
Additionally, instead of simply using the makespan value, 
which is less informative, we divide it by an underestimated 
minimum makespan $\max_{1 \leq i \leq n} \|x_i^I - x_i^G\|$.

\tsmsag\ is implemented in Python. 
The \tsmilp\ implementation is partially based on the Java code from~\cite{YuLav16TRO},
which utilizes the Gurobi ILP solver~\cite{gurobi}.
Since the robot radius $r$ is relative, we fix $r = 1$. 
To capture denseness of the setup, we introduce $2r + \delta$ as 
the minimum distance between two robots allowed when generating $X^I$ and $X^G$.
The region where the robots may appear is limited
to a circle of radius $1.5$ times the radius of the smallest circle 
that $n$ discs with radius $r + \delta / 2$ could fit in~\cite{specht2016the}.
The edge length of hexagonal grids is $\ell = 4/\sqrt{3}$.
The scaling factor $\lambda = \ell / (\delta / 2 + r)$. 
Our evaluation mainly focuses on $d = 0$ since the shift step 
is straightforward.
An example instance is provided in Fig.~\ref{fig:evaluation_setup}.

\begin{figure}[t]
    \vspace*{2mm}
    \centering
    \begin{tabular}{ccc}
        \includegraphics[width = 0.4 \columnwidth, trim={795 305 670 380}, clip]{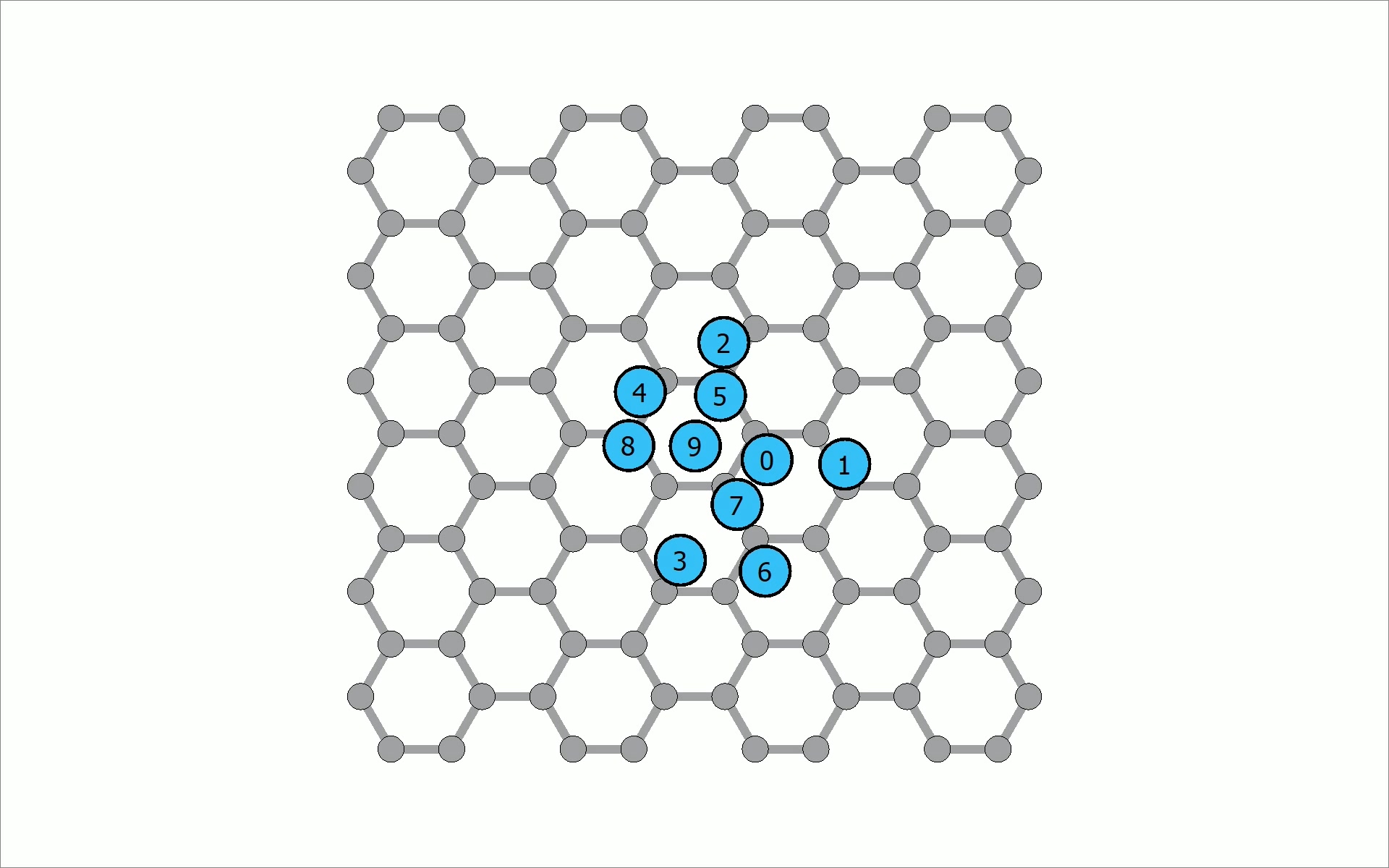}
        &&
        \includegraphics[width = 0.4 \columnwidth, trim={795 315 670 370}, clip]{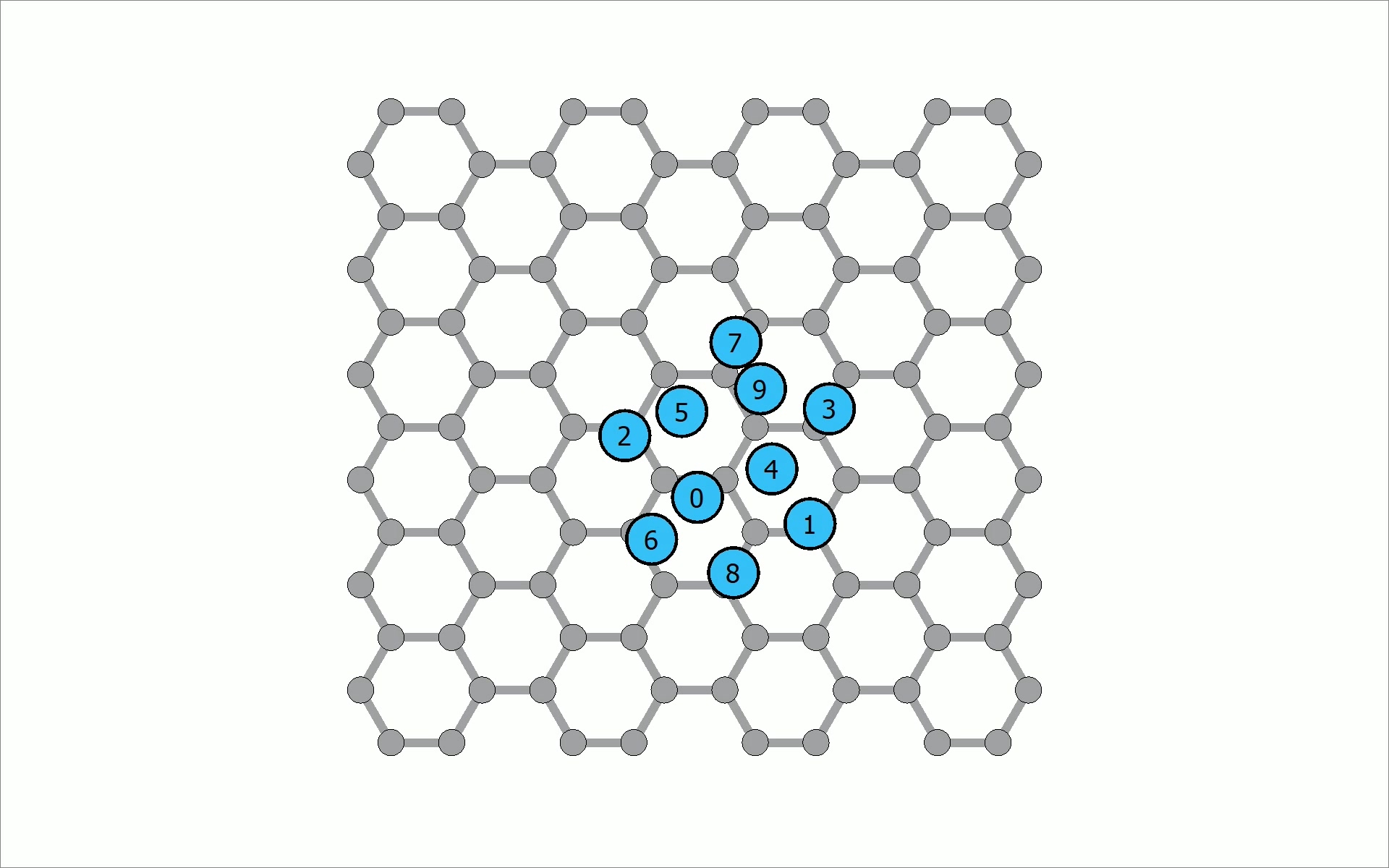}
    \end{tabular}
    \caption{An example of a randomly generated \cmpp\ with 
    $n = 10, \delta = 0, d = 0$. 
    The left and right figure indicates $X^I$ and $X^G$, respectively. 
    Part of the hexagonal graph generated by \esrc\ is rendered in the background.}
    \label{fig:evaluation_setup}
\end{figure}

Problem instances are generated by assiging arbitrary labeling 
to uniformly randomly sampled robot configurations in continuous space. 
For each set of parameters $(n, \delta, d)$, 
$10$ instances are attempted and the average is taken. 
A data point is only recorded if all 10 instances are completed successfully. 
All experiments were executed on a 
Intel\textsuperscript{\textregistered} Core\textsuperscript{TM} i7-6900K CPU 
with 32GB RAM at 2133MHz.

\subsection{\tsmsag}\label{sec:evaluate_sag}

To observe the basic behavior of \tsmsag, we fix $\delta = d = 0$ and evaluate 
\tsmsag\ on both hexagonal grids and square grids by increasing $n$. 
For square grids, $l = 4/\sqrt{2}, \lambda = 2 / (\delta / 2 + r)$.
The result is in Fig.~\ref{fig:evaluate_sag}. 
From the computational time aspect (see Fig.~\ref{fig:evaluate_sag}~(a)),
\tsmsag\ is able to solve problems with $1000$ robots within $800$ seconds.
On the optimality side (see Fig.~\ref{fig:evaluate_sag}~(b)), 
because of the simple implementation and the dense setting,
the optimality ratio of solutions generated by \tsmsag\ is relatively large.
For example, 
the makespan for hexagonal grid is $75$ times 
the underestimated makespan when $n = 20$. 
We can observe, however, that the ratio between the solution makespan and 
the underestimated makespan flattens out as $n$ grows, confirming
the constant factor optimality property of \esrc. 

We note that the guaranteed optimality of \esrc is in fact sub-linear with
respect to the number of robots, which is clearly illustrated in 
Fig.~\ref{fig:evaluate_sag_linear}. As we can observe, the makespan is 
proportional to the number of vertices on the longer side of the grid, 
which in turn is proportional to $\sqrt{n}$. 

\begin{figure}[htp]
    \centering
    \begin{tabular}{ccc}
    \begin{tikzpicture}[scale=1]
        \begin{axis}[
            xtick= {0, 500, 1000},
            ylabel style = {at={(axis description cs:-0.17,.5)}, align=center},
            xlabel={Number of Robots ($n$)},
            ylabel={Computational Time (sec)},
            xmin=0, xmax=1000, ymin=0, ymax=800
        ]
        \addplot[black,every mark/.append style={fill=black},mark=square*] table [x={N}, y={SAG_TIME}] {\datans};
        \addplot[blue,every mark/.append style={fill=blue},mark=diamond*] table [x={N}, y={SAG_TIME}] {\datan};

        \end{axis}
	\end{tikzpicture}
        &&
    \begin{tikzpicture}[scale=1]
        \begin{axis}[
            xtick= {0, 500, 1000},
            ylabel style = {at={(axis description cs:-0.17,.5)}, align=center},
            xlabel={Number of Robots ($n$)},
            ylabel={Optimality Ratio},
            xmin=0, xmax=1000, ymin=0, ymax=400
        ]
        \addplot[black,every mark/.append style={fill=black},mark=square*] table [x={N}, y={SAG_OPT}] {\datans};
        \addplot[blue,every mark/.append style={fill=blue},mark=diamond*] table [x={N}, y={SAG_OPT}] {\datan};
        \end{axis}
	\end{tikzpicture}
        \\
    {\footnotesize (a)} && {\footnotesize (b)} \\
    \multicolumn{3}{c}{
        \begin{tikzpicture}
            \begin{customlegend}[
                legend columns=2,
                legend style={draw=none},
                legend entries={
                    {\scriptsize Square Grid},
                    {\scriptsize Hexagonal Grid}
                    }]
                \addlegendimage{black,every mark/.append style={fill=black},mark=square*}
                \addlegendimage{blue,every mark/.append style={fill=blue},mark=diamond*}
            \end{customlegend}
        \end{tikzpicture}
    }
    \end{tabular}
    \caption{\label{fig:evaluate_sag} 
        Computational time and optimality ratio of \tsmsag\ versus the number of robots.
    }
\end{figure}
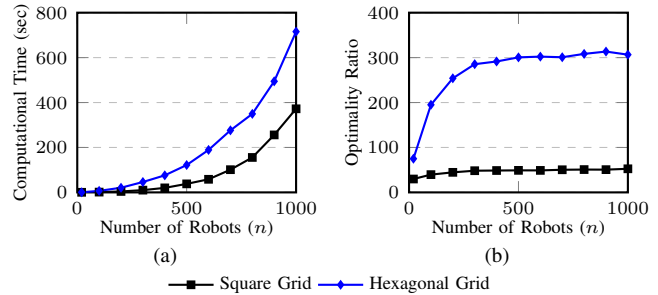

\begin{figure}[htp]
    \centering
    \begin{tikzpicture}[scale=1]
        \begin{axis}[
            legend pos=north west,
            width = 0.46 * 1.61 \linewidth,
            ylabel style = {at={(axis description cs:-0.05,.5)}, align=center},
            xlabel={Number of vertices on the longer side of $G$ ($~\sqrt{n}$)},
            ylabel={Makespan},
            xmin=0, xmax=32, ymin=0, ymax=30000
        ]
        \addplot[black,every mark/.append style={fill=black},mark=square*] table [x={ML}, y={SAG_MAKESPAN}] {\datans};
        \addlegendentry{{\scriptsize Square Grid}}
        \addplot[blue,every mark/.append style={fill=blue},mark=diamond*] table [x={ML}, y={SAG_MAKESPAN}] {\datan};
        \addlegendentry{{\scriptsize Hexagonal Grid}}
        \end{axis}
    \end{tikzpicture}
    \caption{\label{fig:evaluate_sag_linear}
        Makespan of \tsmsag\ versus the number of vertices on the longer side of the grid.
    }
\end{figure}
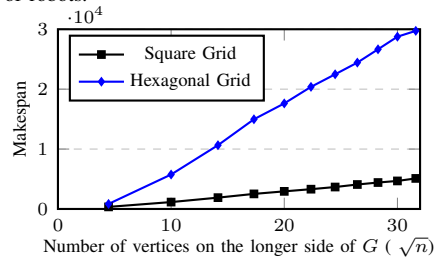

\subsection{\tsmsag\ vs. \tsmilp}

We then compare \tsmsag\ and \tsmilp on hexagonal grids. 
First we let $\delta = d = 0$ and gradually increase $n$.
The performance of the algorithms 
agrees with the expectation (see Fig.~\ref{fig:vary_n}).
\tsmilp\ guarantees the optimal solution on the \mrpp\ sub-problem. 
However, its running time grows quickly and becomes unstable.
It cannot always solve a problem instance in $30$ minutes when $n > 50$. 
Nevertheless, the split heuristic makes \tsmilp\ more applicable when $n$ is large. 
For example, the running time of \tsmilp\ with 4-way split 
is about the same as \tsmsag\ when $n = 700$.
Although \tsmilp\ with split heuristic 
can solve problems with hundreds of robots, 
when $n \geq 800$, only \tsmsag\ could finish
due to its polynomial time complexity.

The makespan of solutions generated by \tsmilp\ is quite low and practical,
since it always provides the optimal solution to \mrpp\ sub-problems.
For instance, when $n = 40$, the makespan of generated solutions 
is $4.93$ times the estimated value. 
The optimality of \tsmilp\ with split heuristics are also close to being optimal.
When $n = 500$, \tsmilp\ with 4-way split provides solutions 
with makespan $4.73$ times the underestimated value.

\begin{figure}[htp]
    \centering
    \begin{tabular}{ccc}
    \begin{tikzpicture}[scale=1]
        \begin{axis}[
            ylabel style = {at={(axis description cs:-0.2,.5)}, align=center},
            xlabel={Number of Robots ($n$)},
            ylabel={Computational Time (sec)},
            xmin=0, xmax=1000, ymin=0, ymax=800
        ]
        \addplot[black,every mark/.append style={fill=black},mark=square*] table [x={N}, y={SAG_TIME}] {\datan};
        \addplot[red,every mark/.append style={fill=red},mark=square*] table [x={N}, y={ILP_TIME}] {\datan};
        \addplot[green,every mark/.append style={fill=green},mark=otimes*] table [x={N}, y={ILP_TIME2}] {\datan};
        \addplot[blue,every mark/.append style={fill=blue},mark=star] table [x={N}, y={ILP_TIME4}] {\datan};
        \end{axis}
	\end{tikzpicture}
        &&
    \begin{tikzpicture}[scale=1]
        \begin{semilogyaxis}[
            ylabel style = {at={(axis description cs:-0.15,.5)}, align=center},
            xlabel={Number of Robots ($n$)},
            ylabel={Optimality Ratio},
            xmin=0, xmax=1000, ymax=400
        ]
        \addplot[black,every mark/.append style={fill=black},mark=square*] table [x={N}, y={SAG_OPT}] {\datan};
        \addplot[red,every mark/.append style={fill=red},mark=square*] table [x={N}, y={ILP_OPT}] {\datan};
        \addplot[green,every mark/.append style={fill=green},mark=otimes*] table [x={N}, y={ILP_OPT2}] {\datan};
        \addplot[blue,every mark/.append style={fill=blue},mark=star] table [x={N}, y={ILP_OPT4}] {\datan};
        \end{semilogyaxis}
	\end{tikzpicture}
        \\
    {\footnotesize (a)} && {\footnotesize (b)} \\
    \multicolumn{3}{c}{
        \begin{tikzpicture}
            \begin{customlegend}[
                legend columns=4,
                legend style={draw=none},
                legend entries={
                    {\scriptsize \tsmsag},
                    {\scriptsize \tsmilp},
                    {\scriptsize \tsmilp-2},
                    {\scriptsize \tsmilp-4}     
                    }]
                \addlegendimage{black,every mark/.append style={fill=black},mark=square*}
                \addlegendimage{red,every mark/.append style={fill=red},mark=square*}
                \addlegendimage{green,every mark/.append style={fill=green},mark=otimes*}
                \addlegendimage{blue,every mark/.append style={fill=blue},mark=star}
            \end{customlegend}
        \end{tikzpicture}
    }
    \end{tabular}
    \caption{\label{fig:vary_n} Performance of \esrc algorithms as $n$ varies.}
\end{figure}
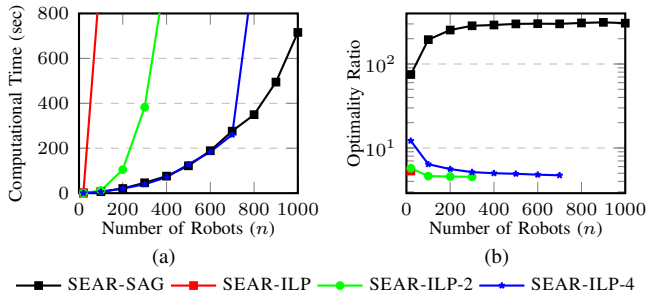

To model different density of robots,
we then fix $n = 30, d = 0$ and increase $\delta$.
With a larger $\delta$, $\lambda$ could be smaller,
which makes robots move less during linear scaling.
Fig.~\ref{fig:vary_e} demonstrates the performance of the methods 
as $\delta$ goes from $0$ to $2.6$ 
(when $\delta \geq 8/\sqrt{3} - 2 \approx 2.62$, no expansion is required).
Here, since $n$ and $|V|$ are not changed,
the size of the \mrpp\ sub-problems remains the same. 
Thus the computational time of the algorithms stays static.
On the other hand, the optimality ratio becomes lower as $\delta$ grows, 
due to the increasement of the underestimated makespan.

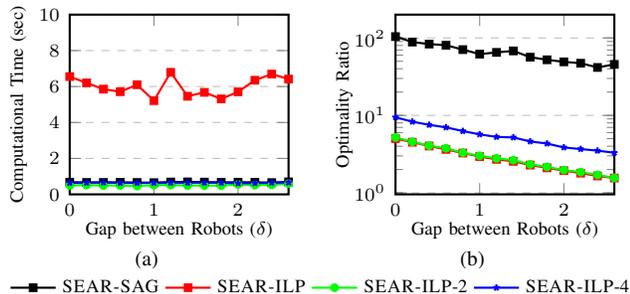
\begin{figure}[htp]
    \centering
    \begin{tabular}{ccc}
    \begin{tikzpicture}[scale=1]
        \begin{axis}[
            ylabel style = {at={(axis description cs:-0.15,.5)}, align=center},
            xlabel={Gap between Robots ($\delta$)},
            ylabel={Computational Time (sec)},
            xmin=0, xmax=2.6, ymin=0, ymax=10
        ]
        \addplot[black,every mark/.append style={fill=black},mark=square*] table [x={EPSILON}, y={SAG_TIME}] {\datae};
        \addplot[red,every mark/.append style={fill=red},mark=square*] table [x={EPSILON}, y={ILP_TIME}] {\datae};
        \addplot[green,every mark/.append style={fill=green},mark=otimes*] table [x={EPSILON}, y={ILP_TIME2}] {\datae};
        \addplot[blue,every mark/.append style={fill=blue},mark=star] table [x={EPSILON}, y={ILP_TIME4}] {\datae};
        \end{axis}
	\end{tikzpicture}
        &&
    \begin{tikzpicture}[scale=1]
        \begin{semilogyaxis}[
            ylabel style = {at={(axis description cs:-0.15,.5)}, align=center},
            xlabel={Gap between Robots ($\delta$)},
            ylabel={Optimality Ratio},
            xmin=0, xmax=2.6, ymax=200
        ]
        \addplot[black,every mark/.append style={fill=black},mark=square*] table [x={EPSILON}, y={SAG_OPT}] {\datae};
        \addplot[red,every mark/.append style={fill=red},mark=square*] table [x={EPSILON}, y={ILP_OPT}] {\datae};
        \addplot[green,every mark/.append style={fill=green},mark=otimes*] table [x={EPSILON}, y={ILP_OPT2}] {\datae};
        \addplot[blue,every mark/.append style={fill=blue},mark=star] table [x={EPSILON}, y={ILP_OPT4}] {\datae};
        \end{semilogyaxis}
	\end{tikzpicture}
        \\
    {\footnotesize (a)} && {\footnotesize (b)} \\
    \multicolumn{3}{c}{
        \begin{tikzpicture}
            \begin{customlegend}[
                legend columns=4,
                legend style={draw=none},
                legend entries={
                    {\scriptsize \tsmsag},
                    {\scriptsize \tsmilp},
                    {\scriptsize \tsmilp-2},
                    {\scriptsize \tsmilp-4}     
                    }]
                \addlegendimage{black,every mark/.append style={fill=black},mark=square*}
                \addlegendimage{red,every mark/.append style={fill=red},mark=square*}
                \addlegendimage{green,every mark/.append style={fill=green},mark=otimes*}
                \addlegendimage{blue,every mark/.append style={fill=blue},mark=star}
            \end{customlegend}
        \end{tikzpicture}
    }
    \end{tabular}
    \caption{\label{fig:vary_e} Performance of \esrc algorithms as $\delta$ varies.}
\end{figure}

For the last part of our experiment, 
we fix $n = 30, \delta = 0$ and change $d$. 
As we can observe in Fig.\ref{fig:vary_d}, 
the solution becomes closer to optimal as $d$ increases, as expected. 

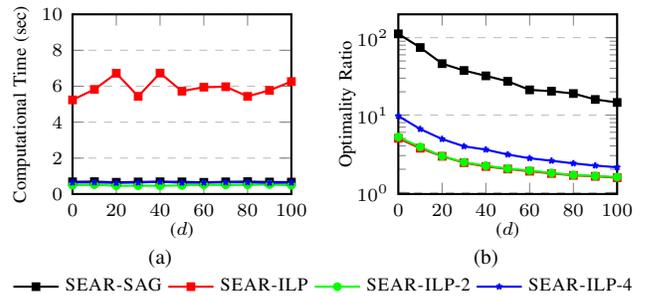
\begin{figure}[htp]
    \centering
    \begin{tabular}{ccc}
    \begin{tikzpicture}[scale=1]
        \begin{axis}[
            ylabel style = {at={(axis description cs:-0.15,.5)}, align=center},
            xlabel={ ($d$)},
            ylabel={Computational Time (sec)},
            xmin=0, xmax=100, ymin=0, ymax=10
        ]
        \addplot[black,every mark/.append style={fill=black},mark=square*] table [x={D}, y={SAG_TIME}] {\datad};
        \addplot[red,every mark/.append style={fill=red},mark=square*] table [x={D}, y={ILP_TIME}] {\datad};
        \addplot[green,every mark/.append style={fill=green},mark=otimes*] table [x={D}, y={ILP_TIME2}] {\datad};
        \addplot[blue,every mark/.append style={fill=blue},mark=star] table [x={D}, y={ILP_TIME4}] {\datad};
        \end{axis}
	\end{tikzpicture}
        &&
    \begin{tikzpicture}[scale=1]
        \begin{semilogyaxis}[
            ylabel style = {at={(axis description cs:-0.15,.5)}, align=center},
            xlabel={ ($d$)},
            ylabel={Optimality Ratio},
            xmin=0, xmax=100, ymax=200
        ]
        \addplot[black,every mark/.append style={fill=black},mark=square*] table [x={D}, y={SAG_OPT}] {\datad};
        \addplot[red,every mark/.append style={fill=red},mark=square*] table [x={D}, y={ILP_OPT}] {\datad};
        \addplot[green,every mark/.append style={fill=green},mark=otimes*] table [x={D}, y={ILP_OPT2}] {\datad};
        \addplot[blue,every mark/.append style={fill=blue},mark=star] table [x={D}, y={ILP_OPT4}] {\datad};
        \end{semilogyaxis}
	\end{tikzpicture}
        \\
    {\footnotesize (a)} && {\footnotesize (b)} \\
    \multicolumn{3}{c}{
        \begin{tikzpicture}
            \begin{customlegend}[
                legend columns=4,
                legend style={draw=none},
                legend entries={
                    {\scriptsize \tsmsag},
                    {\scriptsize \tsmilp},
                    {\scriptsize \tsmilp-2},
                    {\scriptsize \tsmilp-4}
                    }]
                \addlegendimage{black,every mark/.append style={fill=black},mark=square*}
                \addlegendimage{red,every mark/.append style={fill=red},mark=square*}
                \addlegendimage{green,every mark/.append style={fill=green},mark=otimes*}
                \addlegendimage{blue,every mark/.append style={fill=blue},mark=star}
            \end{customlegend}
        \end{tikzpicture}
    }
    \end{tabular}
    \caption{\label{fig:vary_d} Performance of \esrc algorithms as $d$ varies.}
\end{figure}

\begin{figure}[htp]
    \centering
    \vspace*{2mm}
    \begin{tabularx}{1.05\linewidth}{@{}>{\centering\arraybackslash}X>{\centering\arraybackslash}X@{}}
        \includegraphics[keepaspectratio, height= 35mm]{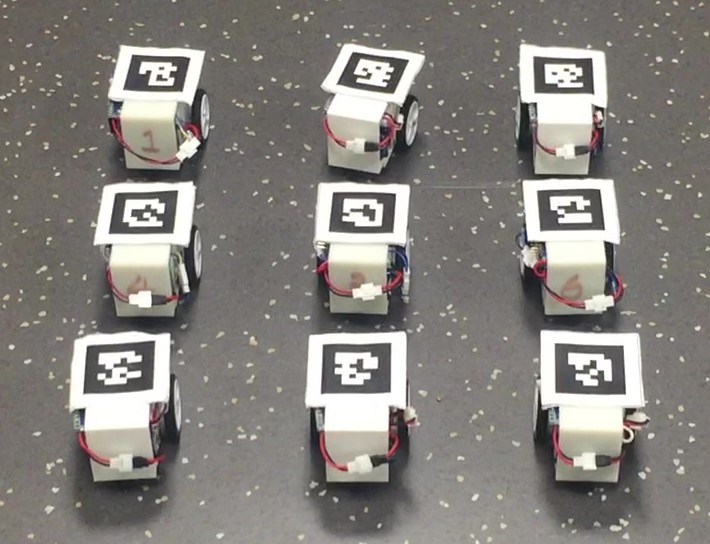}
        &
        \includegraphics[keepaspectratio, height= 35mm]{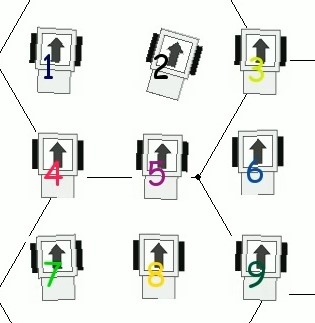}\\
        {\footnotesize (a)} & {\footnotesize (b)} \\
    \end{tabularx}
    \caption{The hardware experimentation. 
    (a) The paths generated by \esrc\ are executed by differential drive cars. 
    (b) Part of a Python-based UI, which is used to monitor and control the vehicles.}
    \label{fig:hardware_setup}
\end{figure}

Both software simulation and hardware experimentation 
of the proposed algorithms can be found in the video attachment. 
The hardware implementation (details in Fig.~\ref{fig:hardware_setup}) 
is based on the microMVP platform~\cite{YuHanTanRus17ICRA}.

\section{Conclusion}\label{sec:conclusion}

In this paper, we proposed the \esrc solution pipeline for tackling 
continuous multi-robot path planning (\cmpp) problem in the absence of obstacles. 
We show that \esrc-based algorithms provide strong theoretical guarantees. With 
an alternative sub-routine for the discrete \mrpp instance, \esrc also proves to
be practical, as confirmed by simulation and hardware experiment. 

The current iteration of the \esrc-based algorithms has focused on 
theoretical guarantees; many additional improvements are possible. 
For example, the expand and assign steps can be done in an adaptive 
manner that avoids uniform scaling, which should greatly boost 
the optimality of the resulting algorithm. Additionally, the current 
\esrc pipeline only provide ad-hoc support for robots differential constraints.
In an extension to the current work, we plan to further improve the efficiency 
and optimality of \esrc and add generic support for robots with differential 
constraints including aerial vehicle swarms. 

%

{\small
\bibliographystyle{IEEEtran}
\bibliography{packed_mrpp.bbl}
}

\end{document}